\documentclass[a4paper,12pt]{article}

\usepackage{algorithm}
\usepackage{algpseudocode}
\usepackage{amssymb}
\usepackage{amsmath}
\usepackage{amsthm}

\newtheorem{theorem}{{\bf Theorem}}
\newtheorem{lemma}{{\bf Lemma}}
\newtheorem{corollary}{{\bf Corollary}}
\newtheorem{example}{{\bf Example}}
\newtheorem{definition}{{\bf Definition}}

\title{Committees providing EJR can be computed efficiently}

\author{Luis S\'anchez-Fern\'andez\\Universidad Carlos III de Madrid
  \and Edith Elkind\\University of Oxford \and Martin
  Lackner\\University of Oxford}

\begin{document}

\maketitle

\begin{abstract}

We identify a whole family of approval-based multi-winner voting rules
that satisfy PJR. Moreover, we identify a subfamily of voting rules
within this family that satisfy EJR. All these voting rules can be
computed in polynomial time as long as the subalgorithms that
characterize each rule within the family are polynomial. One of the
voting rules that satisfy EJR can be computed in $O(n m k)$.

\end{abstract}

\section{Introduction}

The use of axioms in social choice theory dates back to the works of
Arrow~\cite{introSCW}. When the goal of the election is to select two
or more candidates (i.e. a committee), one importat type of axioms are
those that try to identify which requirements must comply a {\it
  representative} set of winners.

In the context of approval-based multi-winner elections this topic has
been adressed by Aziz {\it et al.} in~\cite{aziz:aaai15} (the journal
paper version is~\cite{aziz:scw}). They proposed two axioms to capture
the idea of representation: justified representation (JR) and extended
justified representation (EJR). JR establishes requirements on when a
set of voters deserves a representative while EJR establishes
requirements on when a set of voters deserves several
representatives. Aziz {\it et al.} analysed with these axioms several
well known multi-winner voting rules. The only voting rule that they
found that satisfies EJR is the Proportional Approval Voting
(PAV). Unfortunately, Aziz {\it et al.}~\cite{AGG+14a} and Skowron
{\it et al.}~\cite{skowron2016finding} have proved that computing PAV
is NP-complete. No voting rule computable in polynomial time that
satisfies EJR has been found so far.

S\'anchez-Fern\'andez {\it et al.} proposed in~\cite{pjr-aaai} a
relaxation of EJR that they called proportional justified
representation (PJR). They show that the Greedy Monroe rule (which can
be computed in polynomial time) satisfies PJR if the target committee
size divides the total number of votes. Very soon, Brill {\it et
  al.}~\cite{brill:phragmen} and S\'anchez-Fern\'andez {\it et
  al.}~\cite{2016arXiv160905370S} identified two voting rules that
satisfy PJR in all cases and can be computed in polynomial time.

In this paper we identify a whole family of approval-based
multi-winner voting rules that satisfy PJR. Moreover, we identify a
subfamily of voting rules within this family that satisfy EJR. All
these voting rules can be computed in polynomial time as long as the
subalgorithms that characterize each rule within the family are
polynomial.

\section{Notation and preliminaries}


We consider an approval-based multi-winner election with a set of
voters $N = \{1, \dots, n\}$ and a set of candidates $C= \{c_1, \dots,
c_m\}$. Each voter $i \in N$ submits an approval ballot $A_i \subseteq
C$, which represents the subset of candidates that she approves of. We
refer to the list $\mathcal{A}= (A_1, \dots, A_n)$ of approval ballots
as the {\it ballot profile}. We will consider approval-based
multi-winner voting rules that take as input $(N, C, \mathcal{A}, k)$,
where $k$ is a positive integer that satisfies $k \leq |C|$, and
return a subset $W \subseteq C$ of size $k$, which we call the {\it
  winning set}. We omit $N$ and $C$ from the notation when they are
clear from the context. We say that the {\it exact quota} $q$ is equal
to $n/k$.

For each candidate $c \in C$ we refer to $N_c$ as the set of voters
that approve $c$ ($N_c= \{i \in N: c \in A_i\}$) and to $n_c$ as the
number of voters that approve $c$ ($n_c= |N_c|$).

The maximum of a finite and non-empty set $A$ of real numbers
(represented as $\max A$) is the element of $A$ strictly greater than
all other elements of $A$ (we assume that sets cannot contain
duplicates). By convention we say that $\max \emptyset= 0$.

\begin{definition}

{\bf Extended/proportional justified representation}
Consider a ballot profile $\mathcal{A}= (A_1, \dots, A_n)$ over a
candidate set $C$ and a target committee size $k$, $k \leq |C|$.
Given a positive integer $\ell\in \{1, \ldots, k\}$, we say that a set
of voters $N^*\subseteq N$ is {\em $\ell$-cohesive} if $|N^*| \geq
\ell \frac{n}{k}$ and $|\bigcap_{i \in N^*} A_i| \geq \ell$.  we say
that a set of candidates $W$, $|W| = k$, {\em provides extended
  justified representation (EJR) (respectively, proportional justified
  representation (PJR))} for $(\mathcal{A}, k)$ if for every
$\ell\in\{1, \ldots, k\}$ and every $\ell$-cohesive set of voters $N^*
\subseteq N$ it holds that exists a voter $i$ in $N^*$ such that $|A_i
\cap W| \geq \ell$ (respectively, $|W \cap (\bigcup_{i \in N^*} A_i)|
\ge \ell$). We say that an approval-based voting rule {\em satisfies
  extended justified representation (EJR) (respectively, proportional
  justified representation (PJR))} if for every ballot profile
$\mathcal{A}$ and every target committee size $k$ it outputs a
committee that provides EJR (respectively, PJR) for $(\mathcal{A},
k)$.
\end{definition}

\section{Intuition}

Aziz {\it et al.}~\cite{aziz:scw} proposed a rule, that they called
HareAV, that can be seen as an extension of the largest remainders
apportionment method (with Hare quota) to approval-based multi-winner
elections. HareAV is an iterative method in which at each step the
most approved candidate that has not yet been elected is added to the
set of winners, and at most $\lceil \frac{n}{k} \rceil$ of the votes
that approve the elected candidate are also removed from the election
(when the number of votes that approve the candidate is less than or
equal to $\lceil \frac{n}{k} \rceil$ all such votes are removed from
the election). Which particular votes are removed is left open.

Aziz {\it et al.} dicuss in~\cite{aziz:scw} whether HareAV satisfies
EJR or not. They show an example in which HareAV fails EJR for some
way of breaking intermediate ties. However, they were unable to
construct an example where HareAV fails EJR for all ways of breaking
intermediate ties. Based on this they say: ``we now conjecture that it
is always possible to break intermediate ties in HareAV so as to
satisfy EJR''. Unfortunately, there are elections for which HareAV
fails EJR regardless of the tie-breaking rule, as the following
example, taken from S\'anchez-Fern\'andez {\it et al.}~\cite{pjr-aaai}
shows:

\begin{example}\label{ex:monroe}
Let $n=10$, $k=7$, $C=\{c_1, \dots, c_8\}$.
Suppose that $A_i=\{c_i\}$ for $i=1, \dots, 4$ and
$A_i=\{c_5, c_6, c_7, c_8\}$ for $i=5, \dots, 10$.
Let $\ell = 4$. Then $\ell\cdot\frac{n}{k} =\frac{40}{7} < 6$,
so the set of voters $\{5, 6, 7, 8, 9, 10\}$ form a 4-cohesive group of voters.
However, under HareAV 
only three candidates from 
$\{c_5, c_6, c_7, c_8\}$
will be selected, regardless of the tie-breaking rule used. 
\end{example}

Despite this negative result, the underlying idea of removing part of
the votes in successive iterations offers great flexibility: many
different ways of choosing which votes are removed at each iteration
can be conceived. Based on this, we decided to explore the possibility
of tweaking HareAV in pursuit of a rule that satisfies EJR and can be
computed in polynomial time.

\subsection{First ideas}

The set of winners selected by HareAV in the above example fails to
provide EJR because at each iteration too many votes are removed. In
this example $\lceil \frac{n}{k} \rceil= \lceil \frac{10}{7} \rceil=
2$. Therefore, after 3 iterations, there are $6 - 3 * 2= 0$ votes left
that approve $\{c_5, c_6, c_7, c_8\}$. If instead we remove exactly
$\frac{10}{7}$ votes at each iteration, after 3 iterations there would
be $6 - 3 * \frac{10}{7} = 12/7$ votes left that approve $\{c_5, c_6,
c_7, c_8\}$. This is more than the number of votes received by any of
the other candidates. Therefore, if we remove $\frac{10}{7}$ at each
iteration, all the candidates in $\{c_5, c_6, c_7, c_8\}$ would be
elected. The problem, of course, is that $\frac{10}{7}$ is not an
integer. To solve this we will allow to remove a {\it fraction} of one
vote. For each voter $i$ we define $f_i^j$ as the fraction of the vote
of voter $i$ that has not been removed from the election after $j$
candidates have been added to the set of winners.  For each voter $i$
we define $f_i^0= 1$. At all time, it must be $0 \leq f_i^j \leq 1$.

One advantage of removing fractions of votes instead of whole votes is
that it allows to treat all the votes that approve the elected
candidate evenly: we do not need to ask ourselves which votes must be
removed and which ones must not. We can, for instance, scale down all
the fractions of the votes that approve the elected candidate by the
same factor, so the remaining number of votes is reduced by $q=
\frac{n}{k}$ units. In fact, this rule is not new. According to Svante
Janson~\cite{2016arXiv161108826J}, it was proposed by the Swedish
mathematicians Lars Edvard Phragm\'en and/or Gustaf Enestr\"om in the
19th century (it seems that the original authorship is not absolutely
clear). We will refer to this rule as {\it
  phragm\'en-STV}. phragm\'en-STV is formally defined in
algorithm~\ref{alg:oh-pjr}.

\begin{algorithm}[htb]
\caption{phragm\'en-STV\label{alg:oh-pjr}}
{\bf Input}: an approval-based multi-winner election $(N, C, \mathcal{A}, k)$\\
{\bf Output}: the set of winners $W$
\algblockdefx[ForEach]{ForEach}{EndForEach}[1]
{\textbf{foreach} #1 \textbf{do}}{\textbf{end foreach}} 
\begin{algorithmic}[1]
\State $q\gets \frac{n}{k}$
\State $W\gets \emptyset$

\For{$j=1$ {\bf to} $k$}
  \ForEach{$c \in C \setminus W$}
    \State $s_c^j\gets \sum_{i: c \in A_i} f_i^{j-1}$
  \EndForEach
  \State $w\gets 
    \underset{c \in C \setminus W}{\textrm{argmax}} \ s_c^j$
  \State $W\gets W \cup \{w\}$
  \ForEach{$i \in N \setminus N_w$}
    \State $f_i^j\gets f_i^{j-1}$
  \EndForEach    
  \ForEach{$i \in N_w$}
    \If{$s_w^j \leq q$}
      \State $f_i^j\gets 0$
    \Else
      \State $f_i^j\gets f_i^{j-1} \frac{s_w^j - q}{s_w^j}$
    \EndIf
  \EndForEach    
\EndFor
\State \textbf{return} $W$
\end{algorithmic}
\end{algorithm}

Unfortunately, this rule does not satisfy EJR. Here is one example in
which this rule fails EJR:

\begin{example}
\label{ex2}
Let $n/k=120$, $k=18$, $C=\{c_1, \dots, c_{21}\}$. Voters cast the
following votes:

\begin{itemize}

\item

120 voters approve $\{c_1, c_2, c_5\}$

\item

120 voters approve $\{c_1, c_2, c_6\}$

\item

122 voters approve $\{c_5, c_7\}$

\item

70 voters approve $\{c_3, c_4, c_6\}$

\item

50 voters approve $\{c_3, c_4\}$

\item

120 voters approve $\{c_3, c_4, c_8\}$

\item

121 voters approve $\{c_8, c_9\}$

\item

52 voters approve $\{c_7\}$

\item

65 voters approve $\{c_9\}$

\item

For $i=1, \dots, 12$, 110 voters approve $\{c_{9+i}\}$.

\end{itemize}

For this example the set of winners selected by phragm\'en-STV is one
of $c_3$ or $c_4$ plus $\{c_5, ..., c_{21}\}$. However, the 240 voters
that approve $\{c_1, c_2, c_5\}$ or $\{c_1, c_2, c_6\}$ form a
2-cohesive group of voters but none of these voters approves 2 of the
winners.
\end{example}

In the next section we will examine this example in detail to
understand how to tweak further HareAV to make it satisfy
EJR. Observe, however, that it is very easy to prove that
phragm\'en-STV satisfies PJR. Suppose that $N^*$ is an $\ell$-cohesive
group of voters (that is, $|N^*| \geq \ell \frac{n}{k}$ and $|\cap_{i
  \in N^*} A_i| \geq \ell$). At each iteration in which a candidate
that is approved by some voters in $N^*$ is elected at most
$q= \frac{n}{k}$ votes from $N^*$ are removed. So, after $\ell - 1$
candidates that are approved by some voters in $N^*$ have been
selected, at least $\ell \frac{n}{k} - (\ell -1) \frac{n}{k}=
\frac{n}{k}$ votes from $N^*$ remain in the election, and there exists
at least one candidate that is approved by all such votes. But it is
impossible that after all the $k$ candidates have been selected a
group of $\frac{n}{k}$ votes that approve a common candidate still
remains in the election, so we know that at least one more candidate
approved by some voter in $N^*$ must be elected.

In fact, as long as we remove at most $q= \frac{n}{k}$ votes at each
iteration (and when the number of remaining votes that approve the
selected candidate is less than or equal to $q$ all such votes are
removed from the election), and we select always a candiate that is
approved by at least $q$ of the votes remaining when the candidate is
selected whenever such candidate exists, the rule will satisfy
PJR. This motivates the following definition.

\begin{definition}

We say that a voting rule belongs to the PJR-Exact family (this family
of voting rules is characterized by the use of the exact quota as the
(maximum) amount of votes removed from the election at each iteration)
when the following conditions hold:

\begin{enumerate}

\item

The voting rule consists of an iterative algorithm in which at each
iteration one candidate is added to the set of winners $W$.

\item

At each iteration fractions of the votes are removed from the
election. For each voter $i$ we define $f_i^j$ as the fraction of the
vote of voter $i$ that has not been removed from the election after
$j$ candidates have been added to the set of winners.  For each voter
$i$ we define $f_i^0= 1$. At all time, it must be $0 \leq f_i^j \leq
1$, for all $i, j$. For each voter $i$ and for each $j=0, \dots, k-1$,
it must be $f_i^{j+1} \leq f_i^j$.

If at certain iteration $j$ the number of remaining (fractions of the)
votes that approve the selected candidate $c$ is greater than $q=
\frac{n}{k}$, then only $q$ of such votes are removed from the
election (that is, $\sum_{i: c \in A_i} (f^j_i - f^{j-1}_i)= q$). If
at certain iteration $j$ the number of remaining (fractions of the)
votes that approve the selected candidate $c$ is less than or equal to
$q$, then all such votes are removed from the election (that is, for
all $i$ such that $c \in A_i$, it is $f_i^j=0$).

For each voter $i$ that does not approve the candidate $c$ selected at
iteration $j$, it is $f_i^j= f_i^{j-1}$.

\item

At each iteration $j$, if exists at least one not yet elected
candidate $c$ such that it is approved by at least $q$ (fractions of
the) remaining votes, no candidate that is approved by less than $q$
votes can be selected.

\end{enumerate}

\end{definition}

\begin{theorem}
All the voting rules that belong to the PJR-Exact family satisfy PJR.
\end{theorem}

Then, the next question is: can we find a subfamily of voting rules in
PJR-Exact that satisfy EJR? We address this question in the next
section.

\subsection{Next ideas}

First of all, lets analyze in detail what happens in the first three
iterations of phragm\'en-STV for example~\ref{ex2}.

\paragraph{Iteration 1} The first candidate selected is the most
approved. Such candidate is $c_5$, that is approved by 242 voters. For
each voter that approves $c_5$ (that is, the voters that approve
$\{c_1, c_2, c_5\}$ and the voters that approve $\{c_5, c_7\}$) it is
$f_i^1= f_i^0 \frac{242 - 120}{242}= \frac{122}{242}$. For all the
other voters it is $f_i^1= 1$.

To analyze what has happened in this iteration, lets introduce first
some ideas. First of all, for each not yet elected candidate $c$, we
can establish an upper bound on the possible values of $\ell$, for
which an $\ell$-cohesive group of voters can exist such that all the
members in the group approve $c$. Lets call such bound the {\it
  dissatisfaction level} of such candidate. One possible way (although
a not accurate enough) of defining such bound is $\lfloor \frac{k
  n_c}{n} \rfloor$.

For instance, we can establish the following bounds: $2= \lfloor
\frac{120 + 120}{120} \rfloor$ for $c_1$ and $c_2$, $2= \lfloor
\frac{120 + 122}{120} \rfloor$ for $c_5$ and $1= \lfloor
\frac{122+52}{120} \rfloor$ for $c_7$.

Then, for each voter $i$ we can also establish an upper bound on the
possible values of $\ell$, for which voter $i$ can be a member of a
$\ell$-cohesive group of voters. Lets call this again the {\it
  dissatisfaction level} of the voters. One possible way (although
again a not accurate enough) of defining such bound is the maximum
value of the dissatisfaction levels of the candidates in $A_i$. To
reach that bound would require that the maximum dissatisfaction level
$\ell$ of a candidate in $A_i$ is reached by $\ell$ of such candidates
and that all that $\ell$ candidates are (mostly) approved by the same
voters, so this is a rather pessimistic bound.

Taking into account these ideas, at iteration 1 things seem to be
going well. Observe that for all the voters that approve $\{c_1, c_2,
c_5\}$ and all the voters that approve $\{c_5, c_7\}$, the
dissatisfaction level is clearly less than or equal to 2 (because all
of $c_1$, $c_2$, $c_5$ and $c_7$ have a dissatisfaction level smaller
than or equal to 2). The EJR axiom requires that certain voter in an
$\ell$-cohesive group approves $\ell$ winners. If a voter $i$ has a
dissatisfaction level of $x$, we can think that it would be desirable
that for each candidate approved by $i$ that is added to the set of
winners no more than $\frac{1}{x}$ units of the vote of $i$ are
removed from the election. This would allow that $x$ candidates
approved by $i$ could be added to the set of winners before the entire
vote of $i$ is exhausted. This is the case of iteration 1, because we
remove $\frac{120}{242}$ from each voter that aproves $c_5$, which is
less that $\frac{1}{2}$, and the dissatisfaction level of the voters
that approve $c_5$ is at most 2.

Observe, however, that we can be more accurate in the computation of
the dissatisfaction levels. First of all, we observe that the
dissatisfaction level of a candidate can decrease during the execution
of a rule. Suppose that the dissatisfaction level of certain candidate
$c$ is initially $\ell$. Suppose that after certain number of
iterations some voters that approve $c$ have already at least $\ell$
of their approved candidates in the set of winners. In that situation
we cannot expect that such voters contribute with $\frac{1}{\ell}$
units of vote when candidate $c$ is elected, so we must take this into
account when computing the dissatisfaction levels.

An example of this is the situation of candidate $c_7$ after
iteration~1. Since $122 + 52= 174 > 120$ voters approve $c_7$ we must
expect that the initial dissatisfaction level of candidate $c_7$ must
be 1. However, after candidate $c_5$ is elected at iteration 1 the 122
voters that approve $\{c_5, c_7\}$ already approve one candidate in
the set of winners. For such voters $f_i^1= \frac{122}{242} <
\frac{1}{1}$ (here $\frac{1}{1}$ is the inverse of the initial
dissatisfacton level of candidate $c_7$). Se we can not count any more
on these voters to compute the dissatisfaction level of $c_7$ (and
therefore the dissatisfaction level of $c_7$ after iteration 1 must be
$\lfloor \frac{52}{120} \rfloor= 0$). We observe also that after
iteration 1 the remaining (fractions of) votes that approve $c_7$ are
$122 \frac{122}{242} + 52 \simeq 113,5 < 120$, so it seems difficult
to guarantee that $c_7$ will be elected at a future iteration. 

Observe also that the decrease of the dissatisfaction level of $c_7$
does not look dangerous with respect to our goal of satisfying EJR. In
fact, the voters that approve $c_7$ constitute a 1-cohesive group of
voters. For such voters the requirements imposed by EJR are satisfied
with the election of candidate $c_5$ because the number of voters that
do not approve at least one candidate (that is, voters that approve
only $c_7$) is 52 which is less that 120.

We observe also that when computing the dissatisfaction level of the
voters we do not need to take care of the dissatisfaction levels of
candidates that have already being added to the set of
winners. Suppose that for voter $i$ after $j$ candidates have been
added to the set of winners $W$, the maximum dissatisfaction level of
candidates in $A_i \cap W$ is $\ell$. If $|A_i \cap W| \geq \ell$,
then this voter satisfies the requirements imposed by EJR if she
belongs to a $\ell$-cohesive group of voters. If $|A_i \cap W| < \ell$
and the dissatisfaction levels of all candidates in $A_i \setminus W$
are strictly smaller than $\ell$, then this voter cannot belong to a
$\ell$-cohesive group of voters\footnote{It may happen that even if
  all the candidates in $A_i \setminus W$ have a dissatisfaction level
  strictly smaller than $\ell$ some of them had initially a
  dissatisfaction level greater than or equal to $\ell$. It is then
  possible that voter $i$ belongs to a $\ell$-cohesive group of
  voters, but in such case, other voters in the cohesive group must
  have at least $\ell$ of their approved candidates in the set of
  winners, and thus the requirements imposed by EJR would be also
  satisfied}. In summary, in the computation of the dissatisfaction
level of each voter $i$ we only need to take into account the
dissatisfaction levels of the candidates in $A_i \setminus W$ and the
number of candidates approved by voter $i$ that have already being
added to the set of winners (that is, $|A_i \cap W|$).

The above discussion motivates the following definitions.

\begin{definition}
\label{def:ell}

Consider a ballot profile $\mathcal{A}= (A_1, \ldots, A_n)$ over a
candidate set $C$ and a target committee size $k$, $k \leq |C|$.
Given a set of candidates $W$, $|W| \leq k$, we will represent the
dissatisfaction level of a candidate $c \in C \setminus W$ with
respect to the set of candidates $W$ as $\ell(c, W)$ and its value is
the highest non-negative integer such that $\ell(c, W)= \lfloor
\frac{k}{n} |\{i: c \in A_i \land |A_i \cap W| < \ell(c, W)\}|
\rfloor$.

When certain rule is used to obtain the set of winners for a given
election, and the rule consists of an iterative algorithm in which at
each iteration one candidate is added to the set of winners, we will
use the notation $\ell_j(c)$ to refer to dissatisfaction level of a
candidate $c$ with respect to the first $j$ ($j=0, \ldots, k$)
candidates selected by the rule. That is, for each $j=0, \ldots, k$,
if $W_j$ is the set of the first $j$ candidates added to the set of
winners by the given rule ($W_0= \emptyset$), and $c \in C \setminus
W_j$, then $\ell_j(c)= \ell(c, W_j)$.

\end{definition}

\begin{definition}
Consider a ballot profile $\mathcal{A}= (A_1, \ldots, A_n)$ over a
candidate set $C$ and a target committee size $k$, $k \leq |C|$. The
set of winners $W$ is computed with certain rule that belongs to the
PJR-Exact family. For $j=0, \ldots, k-1$, let $W_j$ be the set of the
first $j$ candidates added by the rule to the set of winners ($W_0=
\emptyset$). For $j=0, \ldots, k-1$, for each candidate $c \in C
\setminus W_j$ and for each voter $i$ such that $c \in A_i$, the
dissatisfaction level $\ell_j(i,c)$ of voter $i$ if candidate $c$ is
added to the set of winners in the next iteration is $\displaystyle
\ell_j(i,c)= \max_{c' \in A_i \setminus (W_j \cup \{c\})} \ell_j(c')$,
and the minimum fraction of $i$'s vote that should be kept for other
candidates if candidate $c$ is added to the set of winners in the next
iteration is $g_i^j(c)$, where

\begin{displaymath}
g_i^j(c)= \left\{ \begin{array}{ll}
0 & \textrm{if $\ell_j(i,c) \leq |A_i \cap W_j|$}\\
\frac{\ell_j(i,c) - |A_i \cap W_j| - 1}{\ell_j(i,c)} 
& \textrm{if $\ell_j(i,c) > |A_i \cap W_j|$}
\end{array} \right.
\end{displaymath}

For $j=0, \ldots, k-1$, we say that a candidate $c \in C \setminus
W_j$ is in a normal state after $j$ candidates have been added to the
set of winners if the following conditions hold:

\begin{enumerate}

\item

For each voter $i$ such that $c \in A_i$, if $\ell_j(i,c) > |A_i \cap
W_j|$ then it has to be $f_i^j \geq \frac{\ell_j(i,c) - |A_i \cap
  W_j|}{\ell_j(i,c)}$, and

\item

$\displaystyle \sum_{i: c \in A_i} (f_i^j - g_i^j(c)) \geq q$.

\end{enumerate}

\end{definition}

Candidates in {\it normal state} are those candidates that can be
safely added to the set of winners. First of all, for each voter $i$
such that $i$ approves the considered candidate $c$ the
dissatisfaction level of voter $i$ is $\ell_j(i,c)$ and this means
that we wish that voter $i$ will approve at least $\ell_j(i,c)$
winners at the end of the election. We note that in the computation of
$\ell_j(i,c)$ we excluded candidate $c$ because if $c$ is added to the
set of winners we would not need to consider the value of $\ell_j(c)$
in the computation of the dissatisfaction level of voter $i$. Then,
the first condition states that for each voter $i$ that approves $c$
either such voter is already satisfied (this means that the number of
already elected winners that the voter approves is at least
$\ell_j(i,c)$) or the voter has enough fraction of vote left so as to
be able to assign $\frac{1}{\ell_j(i,c)}$ to each of the additional
winners that such voter expects to approve in future. This is the
total number of winners that the voter wish to approve ($\ell_j(i,c)$)
minus the number of already elected winners that the voter approves
($|A_i \cap W_j|$).

The second condition establishes that if candidate $c$ is added to the
set of winners it will be possible to remove $q$ votes from the voters
that approve $c$ in such a way that each voter that needs to approve
additional winnerss could still assign $\frac{1}{\ell_j(i,c)}$ to each
of such winners.

We observe that in example~\ref{ex2} candidate $c_5$ is in normal
state before iteration 1. For the voters that approve $\{c_1, c_2,
c_5\}$ we have $\ell_0(i, c_5)= \ell_0(c_1)= \ell_0(c_2)= 2$ and for
the voters that approve $\{c_5, c_7\}$ we have $\ell_0(i, c_5)=
\ell_0(c_7)= 1$. For the first condition, and for all voters that
approve $c_5$ we have $f_i^0 = 1 \geq \frac{\ell_0(i, c_5) - |A_i \cap
  W_0|}{\ell_0(i, c_5)}= \frac{\ell_0(i, c_5) - 0}{\ell_0(i, c_5)}=
1$. For the second condition we have $\sum_{i: c_5 \in A_i} (f_i^0 -
g_i^0(c_5)) = 120 (1 - \frac{1}{2}) + 122 (1 - 0) \geq 120$. However,
candidate $c_7$ is not in normal state because for candidate $c_7$ we
have $\sum_{i: c_7 \in A_i} (f_i^0 - g_i^0(c_7)) = 122 (1 -
\frac{1}{2}) + 52 (1 - 0) = 113 < 120$.

The following lemma relates the concept of dissatisfaction level with
EJR.

\begin{lemma}
\label{lem:dl-ejr}

Consider a ballot profile $\mathcal{A}= (A_1, \ldots, A_n)$ over a
candidate set $C$ and a target committee size $k$, $k \leq
|C|$. Suppose that for certain candidate set $W$, $|W|= k$ it holds
that the dissatisfaction level of each candidate $c \in C \setminus W$
is $0$. Then $W$ provides EJR.

\end{lemma}

\begin{proof}[Proof]
For the sake of contradiction suppose that for certain ballot profile
$\mathcal{A}= (A_1, \ldots, A_n)$ over a candidate set $C$ and certain
target committee size $k$, $k \leq |C|$, for certain candidate set
$W$, $|W|= k$ it holds that for each candidate $c \in C \setminus W$
it is $\ell(c, W)= 0$, but $W$ does not provide EJR.

Then there exists a set of voters $N^* \subseteq N$ and a possitive
integer $\ell$ such that $|N^*| \geq \ell \frac{n}{k}$ and
$|\bigcap_{i \in N^*} A_i| \geq \ell$ but $|A_i \cap W| < \ell$ for
each $i \in N^*$. Moreover, a candidate must exist such that it is
approved by all the voters in $N^*$ but she is not in $W$. Let $c^*$
be such candidate. Clearly, $N^* \subseteq \{i: c^* \in A_i \land |A_i
\cap W| < \ell\}$, and therefore, $\lfloor \frac{k}{n} |\{i: c^* \in
A_i \land |A_i \cap W| < \ell\}|\rfloor \geq \ell$. This implies that
$\ell(c^*, W) \geq \ell > 0$, a contradiction.

\end{proof}

Lets see now what happens in the next two iterations.

\paragraph{Iteration 2} The second candidate selected is $c_8$, that is 
approved by 241 of the remaining votes. For each voter that approves
$c_8$ (that is, the voters that approve $\{c_3, c_4, c_8\}$ and the
voters that approve $\{c_8, c_9\}$) it is $f_i^2= \frac{241 -
  120}{241} f_i^1 = \frac{121}{241}$.

Candidate $c_8$ is in normal state before iteration 2. For the voters
that approve $\{c_3, c_4, c_8\}$ we have $\ell_1(i, c_8)= \ell_1(c_3)=
\ell_1(c_4)= 2$ and for the voters that approve $\{c_8, c_9\}$ we have
$\ell_1(i, c_8)= \ell_1(c_9)= 1$. For the first condition, and for all
voters that approve $c_8$ we have $f_i^1 = 1 \geq \frac{\ell_1(i, c_8)
  - |A_i \cap W_1|}{\ell_1(i, c_8)}= \frac{\ell_1(i, c_8) -
  0}{\ell_1(i, c_8)}= 1$. For the second condition we have $\sum_{i:
  c_8 \in A_i} (f_i^1 - g_i^1(c_8)) = 120 (1 - \frac{1}{2}) + 121 (1 -
0) \geq 120$.

Things continue to look good.

\paragraph{Iteration 3} The third candidate selected is $c_6$, that is 
approved by 190 of the remaining votes. For each voter that approves
$c_6$ (that is, the voters that approve $\{c_1, c_2, c_6\}$ and the
voters that approve $\{c_3, c_4, c_6\}$) it is $f_i^3= \frac{190 -
  120}{190} f_i^2 = \frac{70}{190}$. 

Candidate $c_6$ is not in normal state before iteration 3. For the
voters that approve $\{c_1, c_2, c_6\}$ we have $\ell_2(i, c_6)=
\ell_2(c_1)= \ell_2(c_2)= 2$ and for the voters that approve $\{c_3,
c_4, c_6\}$ we have $\ell_2(i, c_6)= \ell_2(c_3)= \ell_2(c_4)=
2$. Candidate $c_6$ is not in normal state because $\sum_{i: c_6 \in
  A_i} (f_i^2 - g_i^2(c_6)) = 120 (1 - \frac{1}{2}) + 70 (1 -
\frac{1}{2}) = 95 < 120$. Moreover, we observe that the election of
candidate $c_6$ causes a big harm to candidates $c_1$ and
$c_2$. Voters that approve $c_1$ and $c_2$ form a 2-cohesive group of
voters. After iteration 3, all such voters approve only one of the
candidates in the set of winners. This means that to satisfy EJR we
need that an additional candidate approved by such voters would have
to be added to the set of winners (this has to be one of $c_1$ or
$c_2$). However, after iteration 3 the remaining votes that approve
$c_1$ and $c_2$ are $120 \frac{122}{242} + 120 \frac{70}{190} \simeq
104,7 < 120$, and therefore it seems difficult to guarantee that one
of $c_1$ or $c_2$ will be added to the set of winners (in fact they
will never be added to the set of winners). We observe that if the
fraction of vote that was removed from the voters that approve $\{c_1,
c_2, c_6\}$ were less than or equal to $\frac{1}{2}$ more than 120
voters that approve $c_1$ and $c_2$ would remain in the election, and
therefore we would be in the safe side.

The situation that has happened in iteration 3 is captured by the
following definitions.

\begin{definition}
Consider a ballot profile $\mathcal{A}= (A_1, \ldots, A_n)$ over a
candidate set $C$ and a target committee size $k$, $k \leq |C|$. The
set of winners $W$ is computed with certain rule that belongs to the
PJR-Exact family. For $j=0, \ldots, k$, let $W_j$ be the set of the
first $j$ candidates added by the rule to the set of winners ($W_0=
\emptyset$). For $j=0, \ldots, k-1$, for each candidate $c \in C
\setminus W_j$ and for each voter $i$ such that $c \in A_i$ let
$\displaystyle \ell_j(i,c)= \max_{c' \in A_i \setminus (W_j \cup
  \{c\})} \ell_j(c')$.

For $j=0, \ldots, k-1$, we say that a candidate $c \in C \setminus
W_j$ is in a starving state after $j$ candidates have been added to
the set of winners if for some voter $i$ such that $c \in A_i$ it is
$\ell_j(i,c) > |A_i \cap W_j|$ but $f_i^j < \frac{\ell_j(i,c) - |A_i
  \cap W_j|}{\ell_j(i,c)}$.

For $j=0, \ldots, k-1$, we say that a candidate $c \in C \setminus W_j$
is in a eager state after $j$ candidates have been added to the set of
winners if it is not in a starving state and $\sum_{i: c \in A_i}
f_i^j \geq q$ but $\displaystyle \sum_{i: c \in A_i} (f_i^j -
g_i^j(c)) < q$, where

\begin{displaymath}
g_i^j(c)= \left\{ \begin{array}{ll}
0 & \textrm{if $\ell_j(i,c) \leq |A_i \cap W_j|$}\\
\frac{\ell_j(i,c) - |A_i \cap W_j| - 1}{\ell_j(i,c)} 
& \textrm{if $\ell_j(i,c) > |A_i \cap W_j|$}
\end{array} \right.
\end{displaymath}

Finally, for $j=0, \ldots, k-1$, we say that a candidate $c \in C
\setminus W_j$ is in a insufficiently supported state after $j$
candidates have been added to the set of winners if it is not in a
starving state and $\sum_{i: c \in A_i} f_i^j < q$.
\end{definition}

\begin{definition}
\label{def:iter}
Consider a ballot profile $\mathcal{A}= (A_1, \ldots, A_n)$ over a
candidate set $C$ and a target committee size $k$, $k \leq |C|$. The
set of winners $W$ is computed with certain rule that belongs to the
PJR-Exact family. For $j=0, \ldots, k$, let $W_j$ be the set of
the first $j$ candidates added by the rule to the set of winners
($W_0= \emptyset$). For $j=0, \ldots, k-1$, for each candidate $c \in
C \setminus W_j$ and for each voter $i$ such that $c \in A_i$ let
$\displaystyle \ell_j(i,c)= \max_{c' \in A_i \setminus (W_j \cup
  \{c\})} \ell_j(c')$. 

For $j=1, \ldots, k$, we say that iteration $j$ is normal if the
following conditions hold:

\begin{enumerate}

\item

The candidate $c$ that is added to the set of winners is in a normal state.

\item 

For each voter $i$ such that $c \in A_i$ and $\ell_{j-1}(i,c) > |A_i
\cap W_{j-1}|$, it must be $f_i^j \geq
\frac{\ell_{j-1}(i,c) - |A_i \cap W_j|}{\ell_{j-1}(i,c)}$.

\end{enumerate}

For $j=1, \ldots, k$, we say that iteration $j$ is insufficiently
supported if the selected candidate is in an insufficiently supported
state and we say that iteration $j$ is not normal if such iteration is
neither a normal iteration nor an insufficiently supported iteration.

\end{definition}

According to these definitions, we observe that for example~\ref{ex2}
candidate $c_6$ was in an eager state before the execution of
iteration 3 and that candidates $c_1$ and $c_2$ are in a starving
state after the execution of iteration 3.

We observe that candidates that after certain iterarion are in an
eager state do not need to last in such state forever. For the
election in example~\ref{ex2} suppose that in the first iteration we
select candidate $c_5$ but that we remove all the 120 votes from the
voters that approve $\{c_5, c_7\}$; then, in the second iteration we
add candidate $c_1$ to the set of winners and we remove all the 120
votes that approve $\{c_1, c_2, c_5\}$. Now candidate $c_6$ is in a
normal state.

As we have already discussed, the selection of a candidate that is in
an eager state at iteration 3 is closely related to the fact that the
set of winners outputted by phragm\'en-STV for example~\ref{ex2} fails to
provide EJR. To avoid this situation we are going to follow a very
drastic\footnote{And maybe antidemocratic.} but simple approach: we
are going to avoid selecting candidates in an eager state as much as
possible (and we are going to see soon that we can do much to avoid
selecting candidates in an eager state).

Normal iterations have a number of nice properties that we are going
to discuss now. 

\begin{lemma}
\label{lem:l1}
Consider a ballot profile $\mathcal{A}= (A_1, \ldots, A_n)$ over a
candidate set $C$ and a target committee size $k$, $k \leq |C|$. The
set of winners $W$ is computed with certain rule that belongs to the
PJR-Exact family. For $j=0, \ldots, k$, let $W_j$ be the set of the
first $j$ candidates added by the rule to the set of winners ($W_0=
\emptyset$). For $j=0, \ldots, k$, and for each voter $i$, let
$\displaystyle \ell_j(i)= \max_{c \in A_i \setminus W_j} \ell_j(c)$.

Fix a value $j$ such that $1 \leq j \leq k$ and suppose that the first
$j$ iterations of the rule are all normal. Then for each voter $i$ it
holds that either $|A_i \cap W_j| \geq \ell_j(i)$ or $f_i^j \geq
\frac{\ell_j(i) - |A_i \cap W_j|}{\ell_j(i)}$.
\end{lemma}

\begin{proof}[Proof]

For $h=0, \ldots, k-1$, for each candidate $c \in C \setminus W_h$ and
for each voter $i$ such that $c \in A_i$ let $\displaystyle
\ell_h(i,c)= \max_{c' \in A_i \setminus (W_h \cup \{c\})}
\ell_h(c')$. Fix a voter $i \in N$. If no candidate approved by voter
$i$ is added to the set of winners in the first $j$ iterations then
the lemma trivially holds. Suppose that some candidates approved by
voter $i$ have been added to the set of winners during the first $j$
iterations but $|A_{i} \cap W_j| < \ell_j(i)$. Let $r$ be the last of
the first $j$ iterations in which a candidate approved by voter $i$ is
added to the set of winners and let $c_r$ be such candidate. Then,
$f_{i}^j= f_{i}^r$. We have already seen that the dissatisfaction
levels of the voters are monotonically non-increasing, and therefore,
it is $\ell_j(i) \leq \ell_{r-1}(i, c_r)$. Moreover, since iteration
$r$ is normal we have that $f_{i}^r \geq \frac{\ell_{r-1}(i, c_r) -
  |A_{i} \cap W_r|}{\ell_{r-1}(i, c_r)}$. We have

\begin{eqnarray*}
f_{i}^j &=& f_{i}^r \geq \frac{\ell_{r-1}(i, c_r) - |A_{i} \cap
  W_r|}{\ell_{r-1}(i, c_r)} = \frac{\ell_{r-1}(i, c_r) - |A_{i} \cap
  W_j|}{\ell_{r-1}(i, c_r)} \\
&\geq& \frac{\ell_j({i}) - |A_{i} \cap
  W_j|}{\ell_j({i})}
\end{eqnarray*}

\end{proof}

\begin{corollary}
\label{cor:cor1}
Consider a ballot profile $\mathcal{A}= (A_1, \ldots, A_n)$ over a
candidate set $C$ and a target committee size $k$, $k \leq |C|$. The
set of winners $W$ is computed with certain rule that belongs to the
PJR-Exact family. For $j=0, \ldots, k$, let $W_j$ be the set of the
first $j$ candidates added by the rule to the set of winners ($W_0=
\emptyset$). Fix a value $j$ such that $1 \leq j \leq k-1$ and suppose
that the first $j$ iterations of the rule are all normal. Then, after
the first $j$ iterations of the rule no candidate in $C \setminus W_j$
can be in starving state.
\end{corollary}

\begin{proof}
For each voter $i$, let $\displaystyle \ell_j(i)= \max_{c \in A_i
  \setminus W_j} \ell_j(c)$. For each candidate $c \in C \setminus
W_j$ and for each voter $i$ such that $c \in A_i$, let $\displaystyle
\ell_j(i, c)= \max_{c' \in A_i \setminus W_j \cup \{c\}}
\ell_j(c')$. Clearly, for each candidate $c \in C \setminus W_j$ and
for each voter $i$ such that $c \in A_i$, it is $\ell_j(i) \geq
\ell_j(i, c)$. Moreover, by lemma~\ref{lem:l1} it is either $|A_i \cap
W_j| \geq \ell_j(i)$ or $f_i^j \geq \frac{\ell_j(i) - |A_i \cap
  W_j|}{\ell_j(i)}$. Therefore, if for certain voter $i$ it is
$\ell_j(i, c) > |A_i \cap W_j|$, then $f_i^j \geq \frac{\ell_j(i) -
  |A_i \cap W_j|}{\ell_j(i)} \geq \frac{\ell_j(i,c) - |A_i \cap
  W_j|}{\ell_j(i,c)}$.

\end{proof}

\begin{lemma}
\label{lem:l2}
Consider a ballot profile $\mathcal{A}= (A_1, \ldots, A_n)$ over a
candidate set $C$ and a target committee size $k$, $k \leq |C|$. The
set of winners $W$ is computed with certain rule that belongs to the
PJR-Exact family. For $j=0, \ldots, k$, let $W_j$ be the set of the
first $j$ candidates added by the rule to the set of winners ($W_0=
\emptyset$). Fix a value $j$ such that $1 \leq j \leq k$ and suppose
that the first $j$ iterations of the rule are all normal. Then, for
each candidate $c \in C \setminus W_j$ such that $\ell_j(c) \geq 1$,
it holds that $\sum_{i: c \in A_i \land |A_i \cap W_j| < \ell_j(c)}
(f_i^j - \frac{\ell_j(c) - |A_i \cap W_j| - 1}{\ell_j(c)}) \geq q$.
\end{lemma}

\begin{proof}[Proof]
For each voter $i$, let $\displaystyle \ell_j(i)= \max_{c' \in A_i
  \setminus W_j} \ell_j(c')$. Clearly, for each voter $i$ such that $c
\in A_i$ it is $\ell_j(i) \geq \ell_j(c)$. By lemma~\ref{lem:l1} for
each voter $i$ such that $|A_i \cap W_j| < \ell_j(i)$ it is $f_i^j
\geq \frac{\ell_j(i) - |A_i \cap
  W_j|}{\ell_j(i)}$. Definition~\ref{def:ell} implies that $|\{i: c
\in A_i \land |A_i \cap W_j| < \ell_j(c)\}| \geq \ell_j(c)
\frac{n}{k}$. We have

\begin{displaymath}
\begin{array}{l}
{\displaystyle \sum_{i: c \in A_i \land |A_i \cap W_j| < \ell_j(c)}
f_i^j - \frac{\ell_j(c) - |A_i \cap W_j| - 1}{\ell_j(c)})} \\
\quad\quad {\displaystyle \geq 
\sum_{i: c \in A_i \land |A_i \cap W_j| < \ell_j(c)}
(\frac{\ell_j(i) - |A_i \cap
  W_j|}{\ell_j(i)} - \frac{\ell_j(c) - |A_i \cap W_j| - 1}{\ell_j(c)})} \\
\quad\quad {\displaystyle \geq 
\sum_{i: c \in A_i \land |A_i \cap W_j| < \ell_j(c)}
(\frac{\ell_j(c) - |A_i \cap
  W_j|}{\ell_j(c)} - \frac{\ell_j(c) - |A_i \cap W_j| - 1}{\ell_j(c)})} \\
\quad\quad {\displaystyle= \sum_{i: c \in A_i \land |A_i \cap W_j| < \ell_j(c)}
\frac{1}{\ell_j(c)} \geq \ell_j(c) \frac{n}{k} \frac{1}{\ell_j(c)} = q}
\end{array}
\end{displaymath}
\end{proof}

Here comes the good news: it is always possible to avoid adding
candidates in an eager state to the set of winners.

\begin{theorem}
\label{theo:no-eager}
Consider a ballot profile $\mathcal{A}= (A_1, \ldots, A_n)$ over a
candidate set $C$ and a target committee size $k$, $k \leq |C|$. The
set of winners $W$ is computed with certain rule that belongs to the
PJR-Exact family. For $j=0, \ldots, k$, let $W_j$ be the set of the
first $j$ candidates added by the rule to the set of winners ($W_0=
\emptyset$). Fix a value $j$ such that $0 \leq j \leq k-1$ and suppose
that the first $j$ iterations of the rule are all normal ($j= 0$ means
that no iteration has been executed yet). Suppose that after the first
$j$ iterations certain candidate $c_1 \in C \setminus W_j$ is in an
eager state. Then, there exists other candidate $c_2 \in C \setminus
W_j$ that is in a normal state.
\end{theorem}

\begin{proof}
For each candidate $c \in C \setminus W_j$ and for each voter $i$ such
that $c \in A_i$ let $\displaystyle \ell_j(i,c)= \max_{c' \in A_i
  \setminus (W_j \cup \{c\})} \ell_j(c')$. Also, for each candidate $c
\in C \setminus W_j$ and each voter $i$ such that $c \in A_i$ let
$g_i^j(c)$ be

\begin{displaymath}
g_i^j(c)= \left\{ \begin{array}{ll}
0 & \textrm{if $\ell_j(i,c) \leq |A_i \cap W_j|$}\\
\frac{\ell_j(i,c) - |A_i \cap W_j| - 1}{\ell_j(i,c)} 
& \textrm{if $\ell_j(i,c) > |A_i \cap W_j|$}
\end{array} \right.
\end{displaymath} 

First of all, we observe that by lemma~\ref{lem:l1} all the
candidates $c$ in $C \setminus W_j$ satisfy the first condition
required to be in normal state, that is, for each voter $i$ such that
$c \in A_i$ if $\ell_j(i,c) > |A_i \cap W_j|$ it has to be $f_i^j \geq
\frac{\ell_j(i,c) - |A_i \cap W_j|}{\ell_j(i,c)}$.

Now, if candidate $c_1$ is in an eager state, this means that
$\sum_{i: c_1 \in A_i} f_i^j \geq q$ but $\sum_{i: c_1 \in A_i} (f_i^j
- g_i^j(c_1)) < q$. Suppose first that $\ell_j(c_1)= 0$. Since for
each $i$ such that $c_1 \in A_i$ and $\ell_j(i, c_1)= 0$ it is
$g_i^j(c_1)= 0$, there should be certain voter $i_1$ such that $c_1
\in A_{i_1}$ and $\ell_j(i_1, c_1) > 0$ (otherwise $\sum_{i: c_1
  \in A_i} (f_i^j - g_i^j(c_1))= \sum_{i: c_1 \in A_i} f_i^j \geq q$,
and $c_1$ would be in a normal state). If $\ell_j(i_1, c_1) > 0$
this means that there must exist a candidate $c_3$ in $C \setminus
W_j$ such that $c_3 \in A_{i_1}$ and $\ell_j(c_3) > 0$. 

If $c_3$ is in a normal state then the theorem holds. By
lemma~\ref{lem:l2} it has to be $\sum_{i: c_3 \in A_i \land |A_i \cap
  W_j| < \ell_j(c_3)} (f_i^j - \frac{\ell_j(c_3) - |A_i \cap W_j| -
  1}{\ell_j(c_3)}) \geq q$ and therefore $\sum_{i: c_3 \in A_i} f_i^j
\geq q$. Thus, $c_3$ cannot be in an insufficiently supported
state. By corollary~\ref{cor:cor1} $c_3$ cannot be in a starving
state. Therefore, if $c_3$ is not in a normal state then $c_3$ has to
be in an eager state. If $c_3$ is in an eager state then it is
$\sum_{i: c_3 \in A_i} (f_i^j - g_i^j(c_3)) < q$. Combining this with
lemma~\ref{lem:l2} we have $\sum_{i: c_3 \in A_i} (f_i^j - g_i^j(c_3))
- \sum_{i: c_3 \in A_i \land |A_i \cap W_j| < \ell_j(c_3)} (f_i^j -
\frac{\ell_j(c_3) - |A_i \cap W_j| - 1}{\ell_j(c_3)}) <0$. But

\begin{displaymath}
\begin{array}{l}
{\displaystyle \sum_{i: c_3 \in A_i} (f_i^j - g_i^j(c_3))
- \sum_{i: c_3 \in A_i \land |A_i \cap W_j| < \ell_j(c_3)} (f_i^j -
\frac{\ell_j(c_3) - |A_i \cap W_j| - 1}{\ell_j(c_3)})} \\
{\displaystyle \quad\quad \geq \sum_{i: c_3 \in A_i \land |A_i \cap W_j| < \ell_j(c_3)} 
(f_i^j - g_i^j(c_3)) - (f_i^j -
\frac{\ell_j(c_3) - |A_i \cap W_j| - 1}{\ell_j(c_3)})} \\
{\displaystyle \quad\quad = \sum_{i: c_3 \in A_i \land |A_i \cap W_j| < \ell_j(c_3)} 
(\frac{\ell_j(c_3) - |A_i \cap W_j| - 1}{\ell_j(c_3)} - g_i^j(c_3))}
\end{array}
\end{displaymath}

This implies that for some voter $i_3$ such that $c_3 \in A_{i_3}$ and
$|A_{i_3} \cap W_j| < \ell_j(c_3)$ it has to be $\frac{\ell_j(c_3) -
  |A_{i_3} \cap W_j| - 1}{\ell_j(c_3)} - g_{i_3}^j(c_3) < 0$. Since
$|A_{i_3} \cap W_j| < \ell_j(c_3)$, it is $\frac{\ell_j(c_3) -
  |A_{i_3} \cap W_j| - 1}{\ell_j(c_3)} \geq 0$. Therefore,
$g_{i_3}^j(c_3) > 0$. It follows that $\ell_j(i_3, c_3) > |A_{i_3}
\cap W_j|$ and that $g_{i_3}^j(c_3) = \frac{\ell_j(i_3, c_3) -
  |A_{i_3} \cap W_j| - 1}{\ell_j(i_3, c_3)}$. But $\frac{\ell_j(c_3) -
  |A_{i_3} \cap W_j| - 1}{\ell_j(c_3)} - g_{i_3}^j(c_3) =
\frac{\ell_j(c_3) - |A_{i_3} \cap W_j| - 1}{\ell_j(c_3)} -
\frac{\ell_j(i_3, c_3) - |A_{i_3} \cap W_j| - 1}{\ell_j(i_3, c_3)} <
0$ implies that $\ell_j(i_3, c_3) > \ell_j(c_3)$. Then, a candidate
$c_4 \in C \setminus W_j$ must exist such that $c_4 \in A_{i_3}$ and
$\ell_j(c_4) > \ell_j(c_3)$.

We can then repeat for candidate $c_4$ and successively the reasoning
we did for $c_3$. At each step either the candidate that we find is in
normal state or there exists another candidate in $C \setminus W_j$
with a strictly higher dissatisfaction level. We note, however, that
dissatisfaction levels are bounded by $k$, and therefore, at some
point we have to find a candidate $c_2$ that is in a normal state.

The proof when $\ell_j(c_1)> 0$ is similar but we can start reasoning
directly as we did with candidate $c_3$. 
\end{proof}

We have now identified our strategy for the {\it EJR-Exact} family
of voting rules: we will run normal iterations as much as possible. If
at certain point no candidate is available that is in a normal state,
this means that all the unelected candidates have to be in an
insufficiently supported state. We then simply run enough
insufficiently supported iterations to complete the required number of
candidates.

\section{The EJR-Exact family of voting rules}

The EJR-Exact family of voting rules is defined in
algorithms~\ref{alg:ejr-lr1} and \ref{alg:ejr-lr2}. 

\begin{algorithm}[!htb]
\caption{The EJR-Exact family (part 1)\label{alg:ejr-lr1}}
Given the algorithms {\it Alg1}, {\it Alg2} and {\it Alg3}
that characterize each rule\\
\noindent\makebox[\linewidth]{\rule{\textwidth}{0.4pt}}
{\bf Input}: an approval-based multi-winner election $(N, C, \mathcal{A}, k)$\\
{\bf Output}: the set of winners $W$\\
\noindent\makebox[\linewidth]{\rule{\textwidth}{0.4pt}}

\algblockdefx[ForEach]{ForEach}{EndForEach}[1]
{\textbf{foreach} #1 \textbf{do}}{\textbf{end foreach}} 
\begin{algorithmic}[1]
\State $W\gets \emptyset$
\ForEach{$c \in C$}
  \State $\ell_0(c)\gets \lfloor \frac{k n_c}{n} \rfloor$
\EndForEach

\For{$j=1$ {\bf to} $k$}
  \State $C^{\textrm{st1}}\gets \emptyset$
  \ForEach{$c \in C \setminus W$}
    \State $x_c\gets 0$
    \ForEach{$i: c \in A_i$}
       \State $\displaystyle \ell_{j-1}(i,c)\gets 
         \max_{c' \in A_i \setminus (W \cup \{c\})} \ell_{j-1}(c')$
       \If{$\ell_{j-1}(i,c) > |A_i \cap W| + 1$}
         \State $x_c\gets x_c + f_i^{j-1} - 
           \frac{\ell_{j-1}(i,c) - (|A_i \cap W| + 1)}{\ell_{j-1}(i,c)}$ 
       \Else
         \State $x_c\gets x_c + f_i^{j-1}$ 
       \EndIf
    \EndForEach 
    \If{$x_c \geq q$}
      \State $C^{\textrm{st1}}\gets C^{\textrm{st1}} \cup \{c\}$
    \EndIf 
  \EndForEach  
\algstore{break:ejr-lr1}
\end{algorithmic}
\end{algorithm}

Each instance of EJR-Exact family is characterized by 3 algorithms:
{\it Alg1}, {\it Alg2} and {\it Alg3}. At each normal iteration first
{\it Alg1} is executed to select which of the candidates in a normal
state will be added to the set of winners. Then {\it Alg2} is executed
to remove $q$ votes from the voters that approve the selected
candidate in such a way that the rules that must be followed in a
normal iteration are respected. Finally {\it Alg3} is executed to
select the candidates that will be added to the set of winners in
insufficiently supported iterations.

\begin{algorithm}[htb]
\caption{The EJR-Exact family (part 2) \label{alg:ejr-lr2}}
\algblockdefx[ForEach]{ForEach}{EndForEach}[1]
{\textbf{foreach} #1 \textbf{do}}{\textbf{end foreach}} 
\begin{algorithmic}[1]
\algrestore{break:ejr-lr1}
  \If{$C^{\textrm{st1}} \ne \emptyset$}
    \State   \Comment{Stage 1: normal iteration} 
    \State Execute {\it Alg1} to select a candidate $w$ from $C^{\textrm{st1}}$
    \ForEach{$i \in N \setminus N_w$}
      \State $f_i^j\gets f_i^{j-1}$
    \EndForEach    
    \State Execute {\it Alg2} to remove $q$ votes from the voters that 
       approve $w$ 
    \Ensure $f_i^j \geq
\frac{\ell_{j-1}(i,w) - |A_i \cap W| - 1}{\ell_{j-1}(i,w)}$
for each voter $i$ such that $w \in A_i$ and $\ell_{j-1}(i,w) > |A_i
\cap W|$ 
    \State $W\gets W \cup \{w\}$
    \ForEach{$c \in C \setminus W$}
      \State $\ell_j(c)\gets \ell_{j-1}(c)$
      \While{$\ell_j(c) > \lfloor 
        \frac{k}{n} |\{i: c \in A_i \land |A_i \cap W| < \ell_j(c)\}|
        \rfloor$}
        \State $\ell_j(c)\gets \ell_j(c) - 1$
      \EndWhile
    \EndForEach
  \Else
    \State \textbf{break}
  \EndIf
\EndFor
\If{$|W| < k$}
  \State   \Comment{Stage 2: insufficiently supported iterations} 
  \State $j\gets |W|$
  \State Execute {\it Alg3} to add $k-j$ candidates from $C \setminus W$ to $W$
\EndIf
\State \textbf{return} $W$
\end{algorithmic}
\end{algorithm}

The execution of a rule in the EJR-Exact family is as follows. First,
for each candidate $c \in C$, its initial dissatisfaction level
$\ell_0(c)$ is computed in lines~2--4 of
algorithm~\ref{alg:ejr-lr1}. Then, a loop is executed (at most) $k$
times to add candidates to the set of winners. At each iteration $j$
of the loop, first for each candidate $c \in C \setminus W$ we check
if such candidate is in a normal state. The candidates that are in a
normal state are stored in $C^{\textrm{st1}}$. This is done in
lines~6--20 of algorithm~\ref{alg:ejr-lr1}.

If $C^{\textrm{st1}}$ is not empty then we run a normal iteration
(lines~22--34 of algorithm~\ref{alg:ejr-lr2}). Algorithm {\it Alg1} is
executed to select one candidate $w$ from $C^{\textrm{st1}}$ that will
be added to the set of winners and then algorithm {\it Alg2} is
executed to remove $q$ votes from the voters that approve candidate
$w$. This has to be done in such a way that the rules established for
normal iterations are respected (see definition~\ref{def:iter}).
Finally, we update the dissatisfaction levels of the candidates.

If $C^{\textrm{st1}}$ is empty then we exit the loop and, if
necessary, run enough insufficiently supported iterations to complete
$W$ (lines~40--42 of algorithm~\ref{alg:ejr-lr2}).

\begin{theorem}
All the voting rules in the EJR-Exact family satisfy EJR.
\end{theorem}

\begin{proof}[Proof]
It is very easy to prove that all the rules in the EJR-Exact family
satisfy EJR. For the sake of contradiction suppose that for certain
ballot profile $\mathcal{A}= (A_1, \ldots, A_n)$ over a candidate set
$C$ and certain target committee size $k$, $k \leq |C|$, the set of
winners $W$ that certain rule that belongs to the EJR-Exact family
outputs does not provide EJR. Then there exists a set of voters $N^*
\subseteq N$ and a possitive integer $\ell$ such that $|N^*| \geq \ell
\frac{n}{k}$ and $|\bigcap_{i \in N^*} A_i| \geq \ell$ but $|A_i \cap
W| < \ell$ for each $i \in N^*$. Moreover, a candidate must exist such
that candidate is approved by all the voters in $N^*$ but she is not
in the set of winners. Let $c$ be such candidate. Observe that
according to definition~\ref{def:ell}, for each iteration $j$ it is
$\ell_j(c) \geq \ell$.

Suppose first that for that election and rule all the iterations are
normal. Then, since at each iteration we remove $q=\frac{n}{k}$ votes
at the end of the election we have removed all the votes from the
election. However, lemma~\ref{lem:l1} says that for each voter $i$ in
$N^*$ it must be $f_i^k \geq \frac{\ell_k(i) - |A_i \cap
  W|}{\ell_k(i)} > 0$, a contradiction.

Suppose then that for that election and rule we have some
insufficiently supported iterations and let $j$ be the first
insufficiently supported iteration and $W_{j-1}$ be the set of the
first $j-1$ candidates that have been added to the set of winners. We
have seen that when we ran the first insufficiently supported
iteration all the candidates are in an insufficiently supported state,
that is, for each $c' \in C \setminus W_{j-1}$ it is $\sum_{i: c' \in
  A_i} f_i^{j-1} < q$. However, lemma~\ref{lem:l2} says that for
candidate $c$ it is $\sum_{i: c \in A_i \land |A_i \cap W_j| <
  \ell_j(c)} (f_i^j - \frac{\ell_j(c) - |A_i \cap W_j| -
  1}{\ell_j(c)}) \geq q$ and therefore it is also $\sum_{i: c \in A_i}
f_i^{j-1} \geq q$ which is again a contradiction.
\end{proof}

We establish now bounds in computational complexity of the instances
of the EJR-Exact family in terms of the number of arithmetic
operations. We note that to take into account the size of the operands
implies to establish bounds on the number of bits required to
represent each $f_i^j$. While this cannot be done in general for all
the instances of the EJR-Exact family, it can be easily done for
particular instaces. 

\begin{theorem}
\label{theo:complex}
Suppose that for certain instance of EJR-Exact the number of
arithmetic operations required to execute Alg1, Alg2 and Alg3 is
bounded, respectively, by $O(o_1)$, $O(o_2)$ and $O(o_3)$ (depending
on the particular instance, such bounds can depend on $n$, $m$, $k$,
and other factors). Then, the number of arithmetic operations required
to execute such instance in the worst case is bounded by $O(nm^2k +
k(o_1+o_2) + o_3)$.

\end{theorem}

\begin{proof}

It is enough to review the cost of each part of the algorithm. We
assume that the profile is stored in a table with a row per voter and
a column per candidate. Each cell of the table is a bit that is set to
$1$ if the corresponding voter approves the corresponding
candidate. We also assume that the value of $|A_i \cap W|$ for each
voter $i$ is stored in an array of $n$ counters (one for each
voter). Initially, the value of all the counters is set to zero. Each
time that a candidate $w$ is added to the set of winners the counters
corresponding to the voters that approve $w$ are incremented. This can
be done in $O(n)$.

The initial dissatisfaction levels computed in lines~2--4 of
algorithm~\ref{alg:ejr-lr1} requires to compute $n_c$ (this can be
done in $O(n)$) for each candidate, so the total cost of computing the
initial dissatisfaction levels is $O(nm)$.

Then we ran at most $k$ times a loop in which the following tasks are
done:

\begin{enumerate}

\item

First, for each candidate $c$, we have to check if such candidate is
in normal state (lines~7--20 of algorithm~\ref{alg:ejr-lr1}). To do
that, for each voter $i$ that approves the candidate under
consideration we have to compute $\ell_{j-1}(i,c)$. This depends on
the number of candidates approved by voter $i$, and therefore, it can
be done in $O(m)$ operations. In summary, lines~7--20 of
algorithm~\ref{alg:ejr-lr1} can be executed in $O(nm^2)$.

\item

In line~23 of algorithm~\ref{alg:ejr-lr2} a candidate in normal state
is selected in $O(o_1)$.

\item

In line~27 of algorithm~\ref{alg:ejr-lr2} $q$ votes are removed from
the election in $O(o_2)$.

\item

In line~28 of algorithm~\ref{alg:ejr-lr2} a candidate $w$ is added to
the set of winners. As we explained before, also the counters that
store the value of $|A_i \cap W|$ for the voters that approve $w$ are
incremented. This can be done in $O(n)$.

\item

In lines~29--34 of algorithm~\ref{alg:ejr-lr2} the dissatisfaction
levels are updated for each candidate. The while loop in lines~31--33
is executed at most $k$ times, because the dissatisfaction levels are
bounded by $k$. Each time that we execute such while loop we need to
compute $|\{i: c \in A_i \land |A_i \cap W| < \ell_j(c)\}|$. This can
be done in $O(n)$ because we assume that the values of $|A_i \cap W|$
and $\ell_j(c)$ have already been computed and stored in
memory. Therefore, the total computational cost of the execution of
lines~29--34 is bounded by $O(nmk)$.

\end{enumerate}

Finally, the cost of executing line~42 of algorithm~\ref{alg:ejr-lr2}
is bounded by $O(o_3)$. Combining all these data we obtain a worst
case bound of $O(nm^2k + nmk^2 + k(o_1+o_2) + o_3)$. 

It is however possible to establish a better bound for the cost of the
computation of the dissatisfaction levels. As we have already seen,
the initial dissatisfaction level of each candidate is bounded by
$k$. Based on this we have established that each time that the while
loop in lines~31--33 of algorithm~\ref{alg:ejr-lr2} is reached, such
loop is executed at most $k$ times. This is because each time we enter
such loop the dissatisfaction level of the candidate under
consideration is decremented and the dissatisfaction levels cannot
fall below 0.

However, a similar reasoning allows us to conclude that for each
candidate the while loop in lines~31--33 of
algorithm~\ref{alg:ejr-lr2} can be entered at most $k$ times during
the whole execution of an EJR-Exact rule. To see why we note that the
dissatisfaction level of a candidate is never incremented, its initial
value is at most $k$, its value cannot fall below $0$ and each time
that the while loop in lines~31--33 of algorithm~\ref{alg:ejr-lr2} is
entered the dissatisfaction level of the candidate under consideration
is decremented by one unit.

We have to be careful here. The condition in line~31 of
algorithm~\ref{alg:ejr-lr2} is checked each time that the loop is
reached. Therefore, for each candidate such condition is checked at
most $2k$ times (once for each iteration of the outer loop from line~5
of algorithm~\ref{alg:ejr-lr1} to line~38 of
algorithm~\ref{alg:ejr-lr2} plus at most $k$ additional times when the
condition holds). The cost of the evaluation of the condition in
line~31 of algorithm~\ref{alg:ejr-lr2} is $O(n)$. The body of the
while loop (line~32 of algorithm~\ref{alg:ejr-lr2}) is executed in
constant time at most $k$ times for each candidate. Therefore, the
total cost of updating the dissatisfaction levels of the candidates
during an execution of an EJR-Exact rule is at most $O(2nmk + mk)=
O(nmk)$ and the total cost of an execution of an EJR-Exact rule is at
most $O(nm^2k + k(o_1+o_2) + o_3)$.

\end{proof}

A very natural instance of the EJR-Exact family is {\it
  EJR-LR-Even}, defined in algorithms~\ref{alg:ejr-ev1},
\ref{alg:ejr-ev2} and~\ref{alg:ejr-ev3}.

\begin{algorithm}[!htb]
\caption{EJR-LR-Even (part 1)\label{alg:ejr-ev1}}
{\bf Input}: an approval-based multi-winner election $(N, C, \mathcal{A}, k)$\\
{\bf Output}: the set of winners $W$\\
\noindent\makebox[\linewidth]{\rule{\textwidth}{0.4pt}}

\algblockdefx[ForEach]{ForEach}{EndForEach}[1]
{\textbf{foreach} #1 \textbf{do}}{\textbf{end foreach}} 
\begin{algorithmic}[1]
\State $W\gets \emptyset$
\ForEach{$c \in C$}
  \State $\ell_0(c)\gets \lfloor \frac{k n_c}{n} \rfloor$
\EndForEach

\For{$j=1$ {\bf to} $k$}
  \State $C^{\textrm{st1}}\gets \emptyset$
  \ForEach{$c \in C \setminus W$}
    \State $s_c^j\gets \sum_{i: c \in A_i} f_i^{j-1}$
    \State $x_c\gets 0$
    \ForEach{$i: c \in A_i$}
       \State $\displaystyle \ell_{j-1}(i,c)\gets 
         \max_{c' \in A_i \setminus (W \cup \{c\})} \ell_{j-1}(c')$
       \If{$\ell_{j-1}(i,c) > |A_i \cap W| + 1$}
         \State $x_c\gets x_c + f_i^{j-1} - 
           \frac{\ell_{j-1}(i,c) - (|A_i \cap W| + 1)}{\ell_{j-1}(i,c)}$ 
       \Else
         \State $x_c\gets x_c + f_i^{j-1}$ 
       \EndIf
    \EndForEach 
    \If{$x_c \geq q$}
      \State $C^{\textrm{st1}}\gets C^{\textrm{st1}} \cup \{c\}$
     \EndIf 
  \EndForEach  
\algstore{break:ejr-ev1}
\end{algorithmic}
\end{algorithm}

\begin{algorithm}[htb]
\caption{EJR-LR-Even (part 2) \label{alg:ejr-ev2}}
\algblockdefx[ForEach]{ForEach}{EndForEach}[1]
{\textbf{foreach} #1 \textbf{do}}{\textbf{end foreach}} 
\begin{algorithmic}[1]
\algrestore{break:ejr-ev1}
  \If{$C^{\textrm{st1}} \ne \emptyset$}          
    \State $w\gets 
      \underset{c \in C^{\textrm{st1}}}{\textrm{argmax}} \ s_c^j$
      \Comment{Stage 1: normal iteration} 
    \ForEach{$i \in N \setminus N_w$}
      \State $f_i^j\gets f_i^{j-1}$
    \EndForEach    
    \ForEach{$i \in N_w$}
       \If{$\ell_{j-1}(i,c) > |A_i \cap W| + 1$}
         \State $f_i^j\gets f_i^{j-1} -  \frac{q}{x_w} (f_i^{j-1} - 
           \frac{\ell_{j-1}(i,c) - (|A_i \cap W| + 1)}{\ell_{j-1}(i,c)})$ 
       \Else
         \State $f_i^j\gets f_i^{j-1} - \frac{q}{x_w} f_i^{j-1}$ 
       \EndIf
    \EndForEach    
    \State $W\gets W \cup \{w\}$
    \ForEach{$c \in C \setminus W$}
      \State $\ell_j(c)\gets \ell_{j-1}(c)$
      \While{$\ell_j(c) > \lfloor 
        \frac{k}{n} |\{i: c \in A_i \land |A_i \cap W| < \ell_j(c)\}| 
        \rfloor$}
        \State $\ell_j(c)\gets \ell_j(c) - 1$
      \EndWhile
    \EndForEach
  \Else
    \State \textbf{break}
  \EndIf
\EndFor
\algstore{break:ejr-ev2}
\end{algorithmic}
\end{algorithm}

\begin{algorithm}[htb]
\caption{EJR-LR-Even (part 3) \label{alg:ejr-ev3}}
\algblockdefx[ForEach]{ForEach}{EndForEach}[1]
{\textbf{foreach} #1 \textbf{do}}{\textbf{end foreach}} 
\begin{algorithmic}[1]
\algrestore{break:ejr-ev2}
\If{$|W| < k$}       
  \For{$j=|W|+1$ {\bf to} $k$} 
  \Comment{Stage 2: insufficiently supported iterations}
    \ForEach{$c \in C \setminus W$}
      \State $s_c^j\gets \sum_{i: c \in A_i} f_i^{j-1}$
    \EndForEach    
    \State $w\gets 
      \underset{c \in C \setminus W}{\textrm{argmax}} \ s_c^j$
    \State $W\gets W \cup \{w\}$
    \ForEach{$i \in N \setminus N_w$}
      \State $f_i^j\gets f_i^{j-1}$
    \EndForEach    
    \ForEach{$i \in N_w$}
      \State $f_i^j\gets 0$
    \EndForEach
  \EndFor
\EndIf
\State \textbf{return} $W$
\end{algorithmic}
\end{algorithm}

\begin{example}
\label{ex3}
The set of winners produced by EJR-LR-Even for the election shown in
example~\ref{ex2} is (in this order) $c_5, c_8$, $c_1$ or $c_2$, $c_3$
or $c_4$, $c_{10}, \ldots, c_{21}, c_6$, and $c_9$.

We illustrate the operation of EJR-LR-Even with the first iteration
for the election shown in example~\ref{ex2}. The selected candidate is
$c_5$. We have already seen that this candidate is initially in a
normal state. $x_{c_5}$ stores the amount of vote that can be removed
from the election while keeping $\frac{1}{\ell_{j-1}(i,c_5)}$ for
additional candidates when necessary, and its value is $120 (1 -
\frac{1}{2}) + 122= 182$.

An important difference between phargm\'en-STV and EJR-LR-Even is that
in EJR-LR-Even we only scale down the fractions of votes that are
contained in $x_{c_5}$, and therefore the fractions of vote that
remain in the election after iteration 1 are as follows: $f_i^1= f_i^0
- \frac{q}{x_{c_5}} (f_i^0 - \frac{\ell_0(i,c_5) - (|A_i \cap W| +
  1)}{\ell_0(i,c_5)})= 1 - \frac{120}{182} (1 - \frac{2 - (0 +
  1)}{2})= \frac{122}{182}$ for the voters that approve $\{c_1, c_2,
c_4\}$, and $f_i^1= f_i^0 - \frac{q}{x_{c_5}} f_i^0= 1 -
\frac{120}{182} 1 = \frac{62}{182}$ for the voters that approve
$\{c_5, c_7\}$.
\end{example}

\begin{lemma}
The number of arithmetic operations required to compute EJR-LR-Even in
the worst case is bounded by $O(nm^2k)$.
\end{lemma}

\section{Some interesting instances of EJR-Exact}

The results that we have obtained in the previous section are quite
surprising because before only the PAV rule was known to satisfy
EJR. In contrast, we can define as many rules as we want by choosing
what we do in Alg1, Alg2 and Alg3. In this section we discuss some
interesting alternatives. This section is organized in two
parts. First we discuss possible alternatives for Alg1 and Alg2 and
then we consider alternatives for Alg3.

\subsection{Alternatives for Alg1 and Alg2}

\subsubsection{Simple EJR or SEJR}

Suppose that we want to compute a committee that provides EJR for
certain ballot profile, candidate set and target committee
size. Suppose that we run some normal iterations to add some
candidates to the set of winners. After such normal iterations we find
that certain candidate has the maximum dissatisfaction level and that
such dissatisfaction level is greater than or equal to 1. Then,
corollary~\ref{cor:cor1} says that such candidate cannot be in a
starving state, lemma~\ref{lem:l2} says that she cannot be in an
inssuficiently supported state and theorem~\ref{theo:no-eager} says
that she cannot be in an eager state. Thus, she is in a normal
state. Observe that we have not needed to look to the votes that
remain in the election. That is, if we choose at each iteration the
candidate with the highest dissatisfaction level (when ties happen we
can choose any of the tied candidates) we know that there exists a way
to run normal iterations. Thus, we do not need to compute the
fractions of votes that remain in the election.

If at certain point we find that all the candidates that have not
already been added to the set of winners have a dissatisfaction level
equal to 0, it is possible that some of them are in a normal state and
others are in an insufficiently supported state. However,
lemma~\ref{lem:dl-ejr} ensures that in that situation the set of
winners will provide EJR and therefore we do not need to worry about
this.

Algorithm~\ref{alg:sejr} describes the subfamily of voting rules that
operate under these ideas. We refer to this subfamily as Simple EJR
(SEJR) because this is the simplest way that we know to compute
committees that provide EJR.

\begin{algorithm}[!htb]
\caption{The SEJR subfamily\label{alg:sejr}}
Given the algorithm {\it Alg3} that characterize each rule\\
\noindent\makebox[\linewidth]{\rule{\textwidth}{0.4pt}}
{\bf Input}: an approval-based multi-winner election $(N, C, \mathcal{A}, k)$\\
{\bf Output}: the set of winners $W$\\
\noindent\makebox[\linewidth]{\rule{\textwidth}{0.4pt}}

\algblockdefx[ForEach]{ForEach}{EndForEach}[1]
{\textbf{foreach} #1 \textbf{do}}{\textbf{end foreach}} 
\begin{algorithmic}[1]
\State $W\gets \emptyset$
\State $j\gets 0$

\ForEach{$c \in C$}
  \State $\ell_0(c)\gets \lfloor \frac{k n_c}{n} \rfloor$
\EndForEach

\State $m_{\ell}\gets \max_{c \in C \setminus W} \ell_0(c)$

\While{$j < k$ and $m_{\ell} > 0$}
  \State $w\gets \underset{c \in C \setminus W}{\mathrm{argmax}} \ \ell_j(c)$
  \State $W\gets W \cup \{w\}$
  \State $j\gets j+1$
  \ForEach{$c \in C \setminus W$}
    \State $\ell_j(c)\gets \ell_{j-1}(c)$
    \While{$\ell_j(c) > \lfloor 
      \frac{k}{n} |\{i: c \in A_i \land |A_i \cap W| < \ell_j(c)\}| 
      \rfloor$}
      \State $\ell_j(c)\gets \ell_j(c) - 1$
    \EndWhile
  \EndForEach
  \State $m_{\ell}\gets \max_{c \in C \setminus W} \ell_j(c)$
\EndWhile
\If{$|W| < k$}
  \State $j\gets |W|$
  \State Execute {\it Alg3} to add $k-j$ candidates from $C \setminus W$ to $W$
\EndIf
\State \textbf{return} $W$
\end{algorithmic}
\end{algorithm}

\begin{example}
\label{ex4}
We consider again the election described in example~\ref{ex2}. Suppose
that we compute the first winners using SEJR (that is, lines~1--18 of
algorithm~\ref{alg:sejr}). Suppose that we break ties first selecting
the most approved candidates and in the second place by lexicographic
order. The initial dissatisfaction levels of the candidates is: 2 for
$c_1, c_2, c_3, c_4, c_5$, and $c_8$; 1 for $c_6, c_7$, and $c_9$ and
0 for $c_{10}, \ldots, c_{21}$. Therefore, we have to choose first one
of $c_1, c_2, c_3, c_4, c_5$, or $c_8$. It suffices to use the first
tie-breaking rule to select candidate $c_5$. Then, the dissatisfaction
level of candidate $c_7$ falls to 0 and the dissatisfaction levels of
the other candidates do not change. Therefore, in the second iteration
we have to choose one of $c_1, c_2, c_3, c_4$, or $c_8$. It suffices
again to use the first tie-breaking rule to select candidate
$c_8$. The dissatisfaction level of candidate $c_9$ falls to 0 and the
dissatisfaction levels of the other candidates do not change. Thus, in
the third iteration we have to choose one of $c_1, c_2, c_3$, or
$c_4$. All these candidates are approved by 240 votes and therefore we
need to make use of the second tie-breaking rule to select $c_1$. The
dissatisfaction levels of $c_2$ and $c_6$ fall to 0. Finally, in the
fourth iteration $c_3$ is selected, and the dissatisfaction level of
$c_4$ falls to 0. All the candidates in $C \setminus W$ have now a
dissatisfaction level of 0, and therefore we exit the SEJR loop. We
may choose freely any 14 of the remaining candidates and the set of
winners will always provide EJR.
\end{example}

\begin{lemma}
\label{theo:complex_sejr}
Suppose that for certain instance of SEJR the number of arithmetic
operations required to execute Alg3 is bounded, by $O(o_3)$. Then,
such instance can be computed in $O(nmk + o_3)$.
\end{lemma}

\begin{proof}

We have seen at the end of the proof of theorem~\ref{theo:complex}
that the total cost of updating the dissatisfaction levels of the
candidates during an execution of an EJR-Exact rule is at most $O(n m
k)$. It follows immediately that the cost of the execution of an
instance of SEJR is bounded by $O(nmk + o_3)$.

\end{proof}

We note that since in SEJR we only make use of integer arithmetic it
is reasonable to assume that arithmetic operations can be done in
constant time.

\begin{corollary}
For any ballot profile, candidate set and target committee size, it is
possible to compute a set of winners that provides EJR in $O(nmk)$.
\end{corollary}

\subsubsection{Minimizing Wasted Votes}

It seems reasonable to desire that Alg2 removes votes in such a way
that it tries to minimize wasted votes. Suppose that for a given
election after running $j$ normal iterations set of winners is
$W_j$. The next iteration is also normal and candidate $w$ is selected
to be added to the set of winners. Now, we have to remove $q$ of the
votes that approve $w$ from the election. Under the idea of trying to
minimize the number of wasted votes we should probably first remove
the votes from voters that have all their approved candidates in the
set of winners (that is, if for a voter $i$ that approves $w$ it is
$A_i \subseteq W_j \cup \{w\}$, then the vote of this voter should be
one of the first removed in the election.

In the second place, we believe that we should remove the votes of
voters that are already satisfied (we say that a voter $i$ is
satisfied if $|A_i \cap (W_j \cup \{w\})| \geq \ell_j(i, w)$), because
for these voters there is no need to add any other of their approved
candidates to get a set of winners that provides EJR.

\subsection{Alternatives for Alg3: EJR-Exact rules as apportionment methods}

Brill {\it et al.}~\cite{brill:dhondt} presents the following analogy
between multi-winner elections and apportionment problems: ``Any
apportionment problem can be seen as a very simple approval voting
instance: all voters approve all the candidates from their chosen
party, and only those.''. Such analogy can be used as a way to
classify multi-winner voting rules according to which party-list
proportional representation system they reduce. 

Basically, the idea is to map any party-list election to an
approval-based multi-winner election. For each list in the original
party-list election $k$ candidates ($k$ is the number of seats that
must be allocated) are created. Then, if list $A$ received $n_A$ votes
in the original party-list election, also $n_A$ voters approve only
all the candidates created for list $A$ in the approval-based
multi-winner election. For additional details we refer
to~\cite{brill:dhondt}.

The party-list proportional representation system to which a
particular instance of EJR-Exact reduces mainly depends on the
algorithm chosen for Alg3. In particular, EJR-LR-Even reduces to
largest remainders (and hence the LR in its name). In fact, any rule
in the EJR-Exact family that uses the same algorithm as EJR-LR-Even
for Alg3 reduces to largest remainders. 

Largest remainders assigns seats to each list in two steps: first, as
many seats as its lower quota; secondly, it assigns the last seats to
the lists with largest remainders after subtracting to each list total
vote as many quotas as such list has received. In the equivalent
multi-winner election the first step is equivalent to normal
iterations for any instance of EJR-Exact. Then, EJR-LR-Even in the
insufficiently supported iterations assingns seats to the candidates
with higher remaining approval votes; this is equivalent to the second
step of largest remainders.
 
Interestingly, there are also rules in EJR-Exact that reduce to
D'Hondt. Since D'Hondt satisfies also lower quota, running normal
iterations until there are no candidates left in normal state and then
using any approval-based iterative rule such that their iterations are
equivalent to D'Hondt iterations in the apportionment scenario, like
ODH~\cite{2016arXiv160905370S},
seq-phragm\'en~\cite{2016arXiv161108826J,brill:phragmen} or RAV
(surveyed by Kilgour in~\cite{kilgour10}) for Alg3 will produce an
instance that reduce to D'Hondt. Due to its simplicity, we will use
RAV to illustrate the idea.

\begin{definition}
{\bf Reweighted Approval Voting} (RAV) RAV is a multi-round
rule that in each round selects a candidate and then reweights the
approvals for the subsequent rounds. Specifically, it starts by
setting $W = \emptyset$. Then in round $j, j= 1, \ldots, k$, it
computes the {\it approval-weight} of each candidate $c$ as:

\begin{displaymath}
\sum_{i: c \in A_i} \frac{1}{1+|W \cap A_i|},
\end{displaymath}

At each iteration, the candidate with largest approval weight is added
to the set of winners.
\end{definition}

\begin{example}
We can consider using SEJR for the initial iterations and then running
RAV. We refer to this as SEJR-RAV. We use again example~\ref{ex2} to
illustrate how SEJR-RAV works. We first run SEJR as described in
example~\ref{ex4} and get $W= \{c_1, c_3, c_5, c_8\}$. In the first
RAV iteration the candidate with largest approval weight is $c_9$. Its
approval weight is $121 \frac{1}{1+|W \cap \{c_8, c_9\}|} + 65
\frac{1}{1+|W \cap \{c_9\}|}= 125,5$. The remaining candidates added
to the set of winners are $c_7$ and $c_{10}, \ldots, c_{21}$.
\end{example}

\section{Acknowledgements}

We are most grateful to Markus Brill who provided us with the pointer
to phragm\'en-STV.


\bibliographystyle{plain}
\bibliography{dhondt}

\end{document}